\documentclass[12pt]{article}
\usepackage{amssymb,amsthm,hyperref}
\newtheorem{thm}{Theorem}[section]

\theoremstyle{remark}
\newtheorem{rem}{Remark}[section]

\newcommand{\vecq}{\mathbf{q}}
\newcommand{\vecw}{\mathbf{w}}
\newcommand{\vecv}{\mathbf{v}}
\newcommand{\vecI}{\mathbf{I}}

\title{
Intrinsic Properties of Conservation-dissipation Formalism of Irreversible Thermodynamics}
\author{
Wen-An Yong\thanks{E-mail: wayong@tsinghua.edu.cn}\\
\small{\textit{Department of Mathematical Sciences}}\\
\small{\textit{Tsinghua University, Beijing 100084, P.R.China}}}
\date{\small{Dedicated to Heinrich Freist\"uhler on the occasion of his 60th birthday}}

\begin{document}
\maketitle{}

\begin{abstract}
This paper proposes four fundamental requirements for establishing PDEs (partial differential equations) modeling irreversible processes. We show that the PDEs derived via the CDF (conservation-dissipation formalism) meet all the requirements. In doing so, we find useful constraints on the freedoms of CDF and point out that a shortcoming of the formalism can be remedied with the help of the Maxwell iteration. It is proved that the iteration preserves the gradient structure and strong dissipativeness of the CDF-based PDEs.
A refined formulation of the second law of thermodynamics is given to characterize the strong dissipativeness, while the gradient structure corresponds to nonlinear Onsager relations. Further advantages and limitations of CDF will also be presented.
\end{abstract}

%\hspace{-0.5cm}\textbf{Keywords: dissipation potential, gradient systems, }
%\small{}\\
%
%\hspace{-0.5cm}\textbf{AMS subject classification:} \small{35L45,
%35B25, 35M20}

\section{Introduction}

As a branch of macroscopic physics, thermodynamics is a theory studying relationships among various apparently unrelated variables or parameters (extensive and intensive) characterizing a thermodynamic system \cite{Ca}. In classical or equilibrium thermodynamics, the parameters (e.g. temperature, energy, volume, pressure, entropy) are time--space-independent constants for a given system (thermostatics) and the main mathematical tool is multivariable calculus. In non-equilibrium situations, the variables are not constant anymore, but functions of time or space. Thus, the most natural form of the expected relations would be of evolution (partial) differential equations (PDEs).

To see what the expected PDEs look like, we follow de Groot and Mazur \cite{GM} and consider the motion of a one-component fluid. By the first law of thermodynamics, the motion obeys the conservation law of energy. Together with the Galilean invariance, the energy law implies the conservation laws of mass and momentum
%(Jou et. al., 1996)
\cite{EIT_book}, which are usually expressed as
\begin{equation}\label{11}
\begin{array}{rlrl}
\partial_t\rho  + \nabla\cdot(\rho\vecv) & = & 0, \\[3mm]
\partial_t(\rho \vecv) + \nabla\cdot(\rho \vecv\otimes \vecv + {\bf P}) & = & 0, \\[3mm]
\partial_t(\rho E) + \nabla\cdot(\rho\vecv E+ {\bf P}\vecv + \vecq) & = & 0
\end{array}
\end{equation}
for electrically neutral fluids without external forces.
Here $\rho$ is the fluid density, $\vecv$ is the velocity, ${\bf P}$ is the pressure tensor, $E = e + |\vecv|^2/2$ with $e$ the specific internal energy, and $\vecq$ represents the heat flux.
Thus, in 3D we have 5 evolution PDEs for 17 unknowns $\rho, \vecv, e, {\bf P}$ and $\vecq$ (when ${\bf P}$ is not symmetric). They are not closed.
Notice that only first-order time and spatial partial derivatives appear in (\ref{11}).

With fluid flows in mind, we may consider a general physical process and assume that its conservation laws are
\begin{equation}\label{12}
\partial_t u+ \sum_{j=1}^3 \partial_{x_j} f_j=0.
\end{equation}
Here $u=u(t,x)\in\mathbb{R}^n$ represents conserved variables like $u =(\rho, \rho\vecv, \rho E)$ in (\ref{11}),
$x=(x_1,x_2,x_3)$, and $f_j$ is the corresponding flux along the $x_j$-direction.
If each $f_j$ is given in terms of the conserved variables, the system (\ref{12}) becomes closed. In this case, the system is considered to be in local equilibrium and $u$ is also referred to as equilibrium variables. However, very often and as in (\ref{11}) $f_j$ depends on some extra variables in addition to the conserved ones. The extra variables characterize non-equilibrium features of the system under consideration and are called non-equilibrium or dissipative variables. Choosing suitable non-equilibrium variables and determining their evolution equations are the fundamental task of irreversible thermodynamics.

There have been no well-accepted rules for choosing the non-equilibrium variables and determining the evolution equations. So far irreversible thermodynamics has had a number of different theories, such as CIT (Classical Irreversible Thermodynamics, \cite{GM}), RT (Rational Thermodynamics, \cite{Truesdell_book}),
IVT (Internal Variables Thermodynamics, \cite{Muschik3}),
RET (Rational Extended Thermodynamics, \cite{Muller_book}), EIT (Extended Irreversible Thermodynamics, \cite{EIT_book, EIT_book2}), GENERIC (General Equation for Non-Equilibrium Reversible-Irreversible Coupling, \cite{GO, OG, BET_book, PKG}), CDF (Conservation-dissipation Formalism, \cite{ZHYY}), EVA (Energetic Variational Approach, \cite{EVA3}),
%{\bf{SHTC (the Symmetric Hyperbolic Thermodynamically Compatible framework, \cite{PR4}),}}
and so on.
Each of them has its own way in choosing the non-equilibrium variables and deriving the corresponding evolution equations. Good choices are expected to lead to simple governing equations which directly reveal physical insights of the process. An exposition of different versions of continuum thermodynamics and comments can be found in \cite{Muschik3} and references cited therein.

No matter how the PDEs are derived, they should meet certain fundamental physical requirements. In this paper, I will propose four such requirements. They are the observability of physical phenomena, time-irreversibility, long-time tendency to equilibrium, and compatibility with possibly existing classical theories. The fundamentality of these requirements will be expounded in the next section. I will show that the PDEs constructed via our CDF \cite{ZHYY} meet all the four requirements.

Based on the first and second laws of thermodynamics, CDF is a learning based method for choosing the non-equilibrium variables and deriving the corresponding evolution equations \cite{ZHYY}. It assumes that certain given conservation laws are known {\it a priori} and has two freedoms. The freedoms are problem-dependent. It will be seen that our verification of the long-time tendency and compatibility can yield substantial constraints on the freedoms. In addition, the CDF-based PDEs are all of first-order. This shortcoming can be remedied, as shown with our compatibility analysis, by applying the Maxwell iteration \cite{Muller_book} to the CDF-based first-order PDEs so that second-order PDEs can be derived systematically. Further advantages and limitations of CDF will also be presented.

It is proved that the derived second-order PDEs preserve the gradient structure and strong dissipativeness of the first-order ones.
Furthermore, a refined formulation of the second law of thermodynamics is given to characterize the strong dissipativeness. On the other hand, the gradient structure \cite{HMie} corresponds to nonlinear Onsager reciprocal relations for scalar processes (described with the first-order PDEs) and for vectorial or tensorial processes (described with the second-order PDEs) \cite{GM}.

The paper is organized as follows. Section 2 presents the four requirements. CDF is reviewed in Section 3 and illustrated with examples in Section 4. Section 5 is devoted to the long-time tendency and compatibility analysis. Further comments are given in Section 6. The final section summarizes the main conclusions of this paper.

\section{Fundamental requirements}
\setcounter{equation}{0}

In this section, I propose four requirements in closing a given system of conservation laws (\ref{12}). They are the observability of physical phenomena, time-irreversibility, long-time tendency to equilibrium, and compatibility with possibly existing well-validated theories.
The fundamentality of these requirements can be expounded as follows.

The first two of these requirements are quite natural because only observable and irreversible phenomena are of our interest in non-equilibrium thermodynamics.
Observable phenomena are insensitive to small changes of experimental conditions.
The physically correct mathematical description should produce results which are robust against any disturbance of initial or boundary data. Mathematically, the robustness usually corresponds to stability of PDEs with respect to initial or boundary data. For first-order evolution PDEs, hyperbolicity is believed to be necessary for the stability \cite{serre}.

To reflect the time-irreversibility, the expected evolution PDEs should not be time-invariant. This is quite different from the Hamilton's equations in classical mechanics. Actually, irreversible terms have been inherently incorporated into non-equilibrium thermodynamics \cite{GM}. Moreover, such an irreversibility should be related closely to the second law of thermodynamics.

The third one just aims at matching the empirical hypothesis of classical thermodynamics \cite{Ca}, meaning that any thermodynamic system will tend to equilibrium after a long enough time.
%{\bf{It is one of the implications of the GENERIC framework \cite{GO, OG, BET_book, PKG} and plays an important role also in other thermodynamic frameworks \cite{EIT_book2}}}.
Here I would like to make a comment on the term "long time". Physically, a long time is never an infinitely long time, because no one is able to look at his experimental result after an infinitely long time. Indeed, when we use the term "long time", what we really mean is a relatively long time. Thus, fast/slow time scales are physically proper notion and should be used in studying long-time behaviors of a thermodynamic system.

The last one is Bohr's Correspondence Principle. It seems indispensable for any generalization of a classical theory to be reasonable. We do not live at Newton's era. For many problems, there are already well-validated and well-accepted theories. New theories are wanted because more details are needed to be understood and more refined experiments are available. The compatibility means that the new theories should be compatible with the well-validated and classical ones. By the way, this requirement is also a part of Hilbert's sixth problem ("The investigations on the foundations of geometry suggest the problem: To treat in the same manner, by means of axioms, those physical sciences in which mathematics plays an important part; ... Further, the mathematician has the duty to test exactly in each instance whether the new axioms are compatible with the previous ones. ...").

\section{Conservation-dissipation formalism}
\setcounter{equation}{0}

Here we review the conservation-dissipation formalism (CDF) developed in \cite{ZHYY} to close the given conservation laws (\ref{12}). This CDF was inspired by the observation in \cite{Y4} that many well-validated PDEs modeling different physical processes all respect the {\it conservation-dissipation principle} proposed therein. The processes include but not limited to chemically reactive flows, kinetic theories (moment closure systems, discrete-velocity kinetic models),
nonlinear optics,
radiation hydrodynamics, traffic flows, thermal non-equilibrium flows \cite{VK}, dissipative relativistic fluid flows \cite{GL1},
axonal transport \cite{Yanhao},
chemical reactions \cite{Y5}, geophysical flows \cite{Jiawei}, and
compressible viscoelastic flows \cite{Y6, Sader}. See \cite{H2, H3} for recent applications of CDF to soft matter physics.
Thanks to the observation above, it is natural to require the closed PDEs to have the same form and obey the same principle.

With the above simple idea in mind,
we choose a dissipative variable $v\in\mathbb{R}^r$ so that the flux $f_j$ in (\ref{12}) can be expressed as
$ f_j =f_j(u, v)$ and seek evolution equations of the form
% for $v=v(t, x)$:
\begin{equation}\label{31}
 \partial_tv+\sum_{j=1}^3 \partial_{x_j}g_j(u,v)=q(u,v).
\end{equation}
This is our constitutive equation to be determined, where $g_j(u,v)$ is the corresponding flux and $q=q(u,v)$ is the nonzero source, vanishing at equilibrium. Together with the conservation laws ({\ref{12}), the
evolution of a non-equilibrium state
is then governed by a system of first-order PDEs in the compact form
\begin{equation}\label{32}
  \partial_tU+\sum_{j=1}^3 \partial_{x_j}F_j(U)=Q(U),
\end{equation}
where
$$
U=\left(\begin{array}{c}u\\v\end{array}\right),\quad F_j(U)=\left(\begin{array}{c}f_j(U)\\g_j(U)\end{array}\right),\quad Q(U)=\left(\begin{array}{c}0\\q(U)\end{array}\right).
$$
Note that not every thermodynamic variable can evolve in such a balance form, while
many classical systems allow such a set of state variables \cite{VK, GL1, Y4}.

For balance laws (\ref{32}), the aforesaid conservation-dissipation principle consists of the following two conditions.
\begin{enumerate}
\item[(i).] There is a strictly concave smooth function $\eta=\eta(U)$, called entropy (density), such that the matrix product $\eta_{UU}F_{jU}(U)$ is symmetric for each $j$ and for all $U=(u, v)$ under consideration.
\item[(ii).] There is a positive definite matrix $M=M(U)$, called dissipation matrix, such that the non-zero source can be written as  $q(U)=M(U)\eta_v(U).$
\end{enumerate}
Here the subscript stands for the corresponding partial derivative, for instance $\eta_v =\frac{\partial \eta}{\partial v}$ and $\eta_{UU}=\frac{\partial^2 \eta}{\partial U^2}$, and $\eta_v(U)$ should be understood as a column vector. Note that the dissipation matrix is not assumed to be symmetric and its positive definiteness means that of the symmetric part $ \frac{M + M^T}{2}. $
This principle is a strengthened version of the {\it structural stability conditions} proposed in \cite{Y1, Y2, Y3} for hyperbolic systems of PDEs with relaxation. It is strictly stronger than the structure condition of RET \cite{Muller_book}, on which further comments can be found in \cite{Y2, YHB}. In particular, the structure condition of RET does not imply the entropy-production inequality (\ref{35}) below.

Balance laws (\ref{32}) together with the conservation-dissipation principle will be referred to as {\it conservation-dissipation formalism} (CDF).
This formalism has two freedoms: the entropy function $\eta=\eta(U)$ and the dissipation matrix $M=M(U)$. They are both functions of the state variable $U$. The former is strictly concave and the latter is positive definite. Except these, no further restriction is imposed on $\eta=\eta(U)$ or $M=M(U)$.
Specific expressions of $\eta(U)$ and $M(U)$ should be problem-dependent.

Here are some simple comments on the conservation-dissipation principle. Condition (i) is the well-known entropy condition for hyperbolic conservation laws \cite{Go, FL_71}. It corresponds to the classical thermodynamics stability. This condition ensures that the first-order system (\ref{32}) is globally symmetrizable hyperbolic and thereby well-posed \cite{Dafermos}. It implies that there is a function $J_j=J_j(U)$ such that
\begin{equation}\label{33}
\eta_U\cdot F_{jU}=J_{jU}.
\end{equation}
This imposes a restriction on the flux $g_j(u,v)$. Moreover, we use equations (\ref{32}) and (\ref{33}) to compute the rate of change of entropy:
\begin{equation}\label{34}
\begin{array}{rl}
\partial_t \eta&=-\sum_{j=1}^3\eta_U \cdot\partial_{x_j}F_j+\eta_v\cdot q\nonumber\\[3mm]
&=-\sum_{j=1}^3\partial_{x_j}J_j+\sigma
\end{array}
\end{equation}
with the entropy production $\sigma=\eta_v\cdot M(U)\eta_v\ge0.$ Here Condition (ii) has been used. Thus, the second law of thermodynamic and thereby the time-irreversibility are respected automatically by the CDF-based system (\ref{32}).

Condition (ii) is a nonlinearization of the celebrated Onsager reciprocal relation for scalar processes \cite{GM}. Indeed, with Condition (ii) the source term in (\ref{32}) can be rewritten as
$$
Q(U) =\mathcal{M}(U) \eta_U(U), \qquad \mbox{with} \quad \mathcal{M}(U) =\left(\begin{array}{cc}0& 0\\ 0& M(U)\end{array}\right).
$$
Therefore, the balance laws possess the gradient structure \cite{HMie} and can rewritten as
$$
\partial_tU+\sum_{j=1}^3 \partial_{x_j}F_j(U)=\mathcal{M}(U) \eta_U(U).
$$
This form relates the change of entropy $\eta_U(U)$ directly to the various irreversible processes that occur in a system \cite{GM}. In other words, the change of entropy
%acts as an entropic force and
drives the irreversible processes, while the latter lead to the change of entropy.

Remark that the dissipation matrix may depend on the non-equilibrium variables as well as the conserved ones, while the usual Onsager relation only allows the dependence on the conserved variables. It was shown in \cite{Y4} why the dissipation matrix $M=M(U)$ must be positive instead of semi-positive definite. In fact, this positive definiteness guarantees that $\eta_v(u, v)=0$ whenever $q(u,v)=0$. This means that the local equilibrium states are those attaining the maximum of the entropy with respect to the non-equilibrium variable.

We conclude this section with the {\it entropy-production inequality} proposed in \cite{Y3}, which follows from the conservation-dissipation principle. Look at the entropy production in (\ref{34}): $\sigma = \eta_v(U)\cdot M(U)\eta_v(U),$
which is the inner product of the thermodynamic fluxes $Q(U)$ and (entropic) forces $\eta_U(U)$.
It is direct to see that
\begin{equation}\label{35}
\begin{array}{rl}
%\sigma=& \eta_U(U)\cdot Q(U)= \eta_v(U)\cdot M(U)M(U)^{-1}M(U)\eta_v(U)\\[3mm]
%\geq & |M(U)\eta_v(U)|^2/\lambda(U)=|Q(U)|^2/\lambda(U)
\sigma=& \eta_U(U)\cdot Q(U)= \eta_v(U)\cdot M(U)\eta_v(U)\\[3mm]
\geq & \lambda(U)|\eta_v(U)|^2 \geq \frac{\lambda(U)}{|M(U)|^2}|M(U)\eta_v(U)|^2 = \frac{\lambda(U)}{|M(U)|^2}|Q(U)|^2,
\end{array}
\end{equation}
where $\lambda(U)$ is the smallest eigenvalue of the
symmetric
positive-definite matrix $[M(U)+ M^T(U)]/2$ and $|\cdot|$ denotes the usual 2-norm of vectors or matrices.
Namely, the conservation-dissipation principle implies the entropy-production inequality (\ref{35})---a
refined formulation of the second law of thermodynamics.

\section{Examples}
\setcounter{equation}{0}

In this section, we illustrate with examples how CDF guides us to choose the dissipative variable $v$ and to determine the corresponding fluxes $g_j(u,v)$.
The interested reader is referred to \cite{ZHYY} for heat conduction in rigid bodies and the doctoral thesis \cite{Yang} for multi-component fluid mixtures.

(4A). The first example is about fluid flows through porous media. The process obeys the conservation law of mass:
\begin{equation}\label{41}
\partial_t\rho +\nabla\cdot(\rho\vecv)=0.
\end{equation}
This equation is not closed since the fluid velocity $\vecv$ is unknown. Our task is to close this equation.

To this end, we introduce a non-equilibrium state variable $\vecw$ with the size of the fluid velocity $\vecv$, but not necessarily the unknown velocity. Thus, the system is characterized with the state variable $(\rho, \rho\vecw)$. This choice of the state variable, instead of $(\rho, \vecw)$, is inspired by the classical calculations \cite{GM}. See also the following calculations.

Then we specify a strictly concave function $\eta=\eta(\rho, \rho\vecw)$ as an entropy density in the CDF. In order to compare with the classical calculations \cite{GM}, we define a generalized specific entropy
$$
s = s(\nu, \vecw) = \nu \eta(1/\nu, \vecw/\nu)
$$
depending on the specific volume $\nu = 1/\rho$ and the non-equilibrium variable $\vecw$.
Recall that an isothermal system in equilibrium usually has a specific entropy $s_0 = s_0(\nu)$.
%depending on the specific volume $\nu = 1/\rho$ and assume that the non-equilibrium system possesses This specific entropy corresponds to the entropy (density) in CDF with the relation
%$$
%\eta = \eta(\rho, \rho\vecw) = \rho s(1/\rho, \vecw).
%$$
It is standard to show that the concavity of $\eta= \eta(\rho, \rho\vecw)$ is equivalent to that of $s=s(\nu, \vecw)$.
%For the convenience of the reader, we prove it here. Fix two states $(\rho_1, \vecw_1)$ and $(\rho_2, \vecw_2)$. Let $\gamma\in[0, 1]$. From the concavity of $s=s(\nu, \vecw)$ it follows that
%$$
%\begin{array}{rl}
%& \eta(\gamma\rho_1 + (1-\gamma)\rho_2, \gamma\rho_1\vecw_1 + (1-\gamma)\rho_2\vecw_2) \\[3mm]
%= & [\gamma\rho_1 + (1-\gamma)\rho_2]s\Big(\frac{1}{\gamma\rho_1 + (1-\gamma)\rho_2}, \frac{\gamma\rho_1\vecw_1 + (1-\gamma)\rho_2\vecw_2}{\gamma\rho_1 + (1-\gamma)\rho_2}\Big)\\[3mm]
%= & [\gamma\rho_1 + (1-\gamma)\rho_2]s\big(\frac{\xi}{\rho_1}+\frac{1-\xi}{\rho_2}, \xi\vecw_1 + (1-\xi)\vecw_2\big) \quad (\mbox{with} \ \xi = \frac{\gamma\rho_1}{\gamma\rho_1 +(1-\gamma)\rho_2})  \\[3mm]
%\geq & [\gamma\rho_1 + (1-\gamma)\rho_2][\xi s\big(\frac{1}{\rho_1}, \vecw_1\big)+(1-\xi) s\big(\frac{1}{\rho_2}, \vecw_2\big)]\\[3mm]
%= & \gamma\rho_1 s\big(\frac{1}{\rho_1}, \vecw_1\big)+(1-\gamma)\rho_2 s\big(\frac{1}{\rho_2}, \vecw_2\big),
%\end{array}
%$$
%that is, $\eta= \eta(\rho, \rho\vecw)$ is concave. Thanks to the symmetry $s(\nu, \vecw)=\nu\eta(1/\nu, \vecw/\nu)$, the same argument also shows that the concavity of $s=s(\nu, \vecw)$ follows from that of $\eta= \eta(\rho, \rho\vecw)$.
In what follows, we will use $s$ instead of $\eta$.

As in \cite{ZHYY}, we introduce a differential operator ${\mathcal D}$ acting on a function $f=f(x, t)$ as  ${\mathcal D}f = \partial_t(\rho f) + \nabla\cdot(\rho\vecv f)$.
Thanks to the continuity equation (\ref{41}), it is immediate to see that ${\mathcal D}f = \rho(\partial_tf + \vecv\cdot\nabla f)$ and thereby is Galilean invariant.
In order to be consistent with the equilibrium thermodynamics, we define the non-equilibrium (isothermal) pressure $\pi$ with
$$
\pi = s_\nu(\nu, \vecw).
$$
Then we use the conservation law (\ref{41}) and the generalized Gibbs relation
$$
ds = \pi d\nu + s_\vecw\cdot d\vecw
$$
to compute the rate of change of entropy:
$$
\begin{array}{rl}
{\mathcal D} s =&\pi\nabla\cdot\vecv + s_\vecw\cdot{\mathcal D}\vecw\nonumber\\[3mm]
= & \nabla\cdot(\pi\vecv)-\vecv\cdot\nabla \pi +s_\vecw\cdot{\mathcal D} \vecw\nonumber\\[3mm]
%%=&-\nabla\cdot(\theta^{-1}\vecq)+s_\vecw\cdot\partial_t \vecw+\vecq\cdot\nabla\theta^{-1}\\
\equiv&-\nabla\cdot \mathbf{J}+\sigma.
\end{array}
$$
Here $\mathbf{J}=-\pi\vecv$ and $\sigma=s_\vecw\cdot{\mathcal D} \vecw - \vecv\cdot\nabla\pi$ is the entropy production.
With this expression of the entropy production, our CDF suggests to choose
$$
\vecv= - s_\vecw(\nu, \vecw)
$$
and
\begin{equation}\label{42}
{\mathcal D} \vecw +\nabla\pi \equiv \partial_t(\rho\vecw) + \nabla\cdot(\rho\vecv\otimes\vecw)+\nabla\pi = -\mathbf{M}(\rho,\rho\vecw)\vecv,
\end{equation}
where $\mathbf{M}=\mathbf{M}(\rho, \rho\vecw)$ is a positive definite matrix. Note that this equation is not Galilean invariant due to
the dissipation proportional to the velocity field $\vecv$.
Equations (\ref{41}) and (\ref{42}) together form a system of first-order PDEs in the form (\ref{32}) with
$$
U= \left(\begin{array}{c}
   \rho\\
  \rho\vecw
  \end{array}\right), \ \,
  \sum_j \partial_{x_j}F_j(U)=\nabla\cdot\left(\begin{array}{c}
 -\rho s_\vecw\\
 -\rho s_\vecw\otimes\vecw + \pi\vecI
      \end{array}\right),\ \,Q(U)= \left(\begin{array}{c}
    0\\
    \mathbf{M}s_\vecw
   \end{array}\right),
$$
where $\vecI$ is the 3$\times$3 identity matrix.

From the above procedure, we see that the non-equilibrium variable $\vecw$ is conjugated to the fluid velocity $\vecv$ with respect to the pre-specified entropy. Thanks to the strict concavity of $s=s(\nu, \vecw)$,
the non-equilibrium variable $\vecw$ can be globally expressed in terms of $\rho$ and $\vecv$ \cite{FL_71}.
%In contrast to EIT where directly derived was the evolution equation of the flux $\vecq$, CDF gives equation (\ref{42}). Notice that EIT does not guarantee the hyperbolicity of its final system consisting of (\ref{41}) and an evolution equation for $\vecq$ \cite{EIT_book}.

At this point, let us mention that the use of such conjugate variables are traditional in thermodynamics \cite{Ca}. Their importance was further elaborated in \cite{RS} and were called main field therein.

For the freedoms, a simple choice is
$$
s(\nu,\vecw)=s_0(\nu) - \frac{|\vecw|^2}{2\alpha}, \quad \mathbf{M}=\frac{\alpha}{\lambda\varepsilon}\vecI .
$$
Here $s_0(\nu)$ is the equilibrium entropy, $\alpha$ is a positive constant with unit kg/K, $\lambda$ is a positive constant with unit $m^3$/kg, and $\varepsilon$ is a positive constant related to the relaxation time. With this choice, we have $\vecv = \vecw/\alpha$ and equation (\ref{42}) reduces to the classical momentum equation with damping:
$$
\partial_t (\rho\vecv) + \nabla\cdot(\rho\vecv\otimes\vecv) + \nabla(\frac{\pi}{\alpha}) = -\frac{1}{\lambda\varepsilon}\vecv.
$$
Thus, equation (\ref{42}) can be regarded as a nonlinear generalization of the momentum equation. \\

(4B). Next example is from \cite{ZHYY} to close system (\ref{11}) for one-component fluid flows.
Referring to the above experience, we introduce two non-equilibrium state variables $\vecw$ and ${\bf C}$ which have the respective sizes of the vector $\vecq$ and tensor ${\bf P}$. Recall that such a system in equilibrium usually has a specific entropy $s_0=s_0(\nu, e)$
with $\nu = 1/\rho$ and $e=E-|\vecv|^2/2$ (the specific internal energy). Now we assume that the non-equilibrium system under consideration possesses a generalized specific entropy
\begin{equation}\label{43}
s = s(\nu, e, \vecw, {\bf C})
\end{equation}
depending on the non-equilibrium variables $(\vecw, {\bf C})$ as well as the classical ones $(\nu, e)$.
This specific entropy corresponds to the entropy density $\eta$ in CDF with the relation
$$
\eta=\eta(\rho,\rho\vecv,\rho E,\rho \vecw,\rho{\bf C})=\rho s(1/\rho, E-|\vecv|^2/2, \vecw, {\bf C}).
$$
%As is well-known,  the concavity of $\eta$ in its arguments is equivalent to that of $s$ in its own arguments.
Accordingly, we define the non-equilibrium temperature $\theta$ and the non-equilibrium
thermodynamic pressure $\pi$ as
\begin{equation}\label{44}
\theta^{-1} :=s_e(\nu, e, \vecw, {\bf C}), \quad \theta^{-1}\pi :=s_\nu(\nu, e, \vecw, {\bf C}).
\end{equation}
For convenience, we exempt the thermodynamic pressure $\pi$ from the stress ${\bf P}$. Namely, set
$$
{\mathcal\tau}={\bf P}-\pi\vecI
$$
which accounts for possible dissipative effects such as viscosity.
Moreover, we have the generalized Gibbs relation
$$
\hbox{d}s=\theta^{-1}\left[\pi\hbox{d}\nu +d e\right]+s_\vecw\cdot\hbox{d}\vecw+s_{\bf C}^T:\hbox{d}{\bf C}.
$$
Here the superscript $T$ denotes the transpose and the colon $:$ stands for the double contraction of two second-order tensors:  $\mathbf{A}:\mathbf{B}=\sum_{i,j}A_{ij}B_{ji}$.

Then we use the conservation laws in (\ref{11}) and the generalized Gibbs relation to compute the rate of change of entropy:
$$
\eta_t+\nabla\cdot(\vecv\eta)=-\nabla\cdot\mathbf{J}+\sigma .
$$
The details can be found in \cite{ZHYY}. Here $\mathbf{J}=\theta^{-1}\vecq$ is the entropy flux and
$$
\sigma=(s_\vecw\cdot\mathcal{D}\vecw+\vecq\cdot\nabla \theta^{-1})+(s_{\bf C}^T:\mathcal{D} {\bf C} -\theta^{-1}{\mathcal\tau}^T:\nabla\vecv)
$$
is the entropy production.
Having this expression of the entropy production, we refer to CDF and choose $\vecq= s_\vecw, ~{\mathcal\tau} = \theta s_{\bf C}$,
and
\begin{equation}\label{45}
  \left(\begin{array}{c}
\partial_t(\rho\vecw)+\nabla\cdot(\rho\vecv\otimes\vecw)+\nabla\theta^{-1}\\[3mm]
\partial_t(\rho{\bf C})+\nabla\cdot(\rho\vecv\otimes{\bf C})-\nabla\vecv
  \end{array}\right)=\mathbf{M}\cdot  \left(\begin{array}{c}
    \vecq\\[3mm]
    \theta^{-1}\mathcal{\tau}
  \end{array}\right)
\end{equation}
with $\mathbf{M}=\mathbf{M}(\rho, e, \vecw, {\bf C})$ positive definite.
Consequently, the final closed system of governing equations consists of the constitutive equations (\ref{45}) and the conservation laws in (\ref{11}).

Up to now, we have not assumed the symmetry of the stress tensor. When it is symmetric, we will take the non-equilibrium tensor ${\bf C}$ to be symmetric. In this case, the previous calculations are still valid but lead to
\begin{equation}\label{46}
  \left(\begin{array}{c}
\partial_t(\rho\vecw)+\nabla\cdot(\rho\vecv\otimes\vecw)+\nabla\theta^{-1}\\[3mm]
\partial_t(\rho{\bf C})+\nabla\cdot(\rho\vecv\otimes{\bf C})-\frac{1}{2}(\nabla\vecv+\nabla\vecv^T)
  \end{array}\right)=\mathbf{M}\cdot  \left(\begin{array}{c}
    \vecq\\[3mm]
    \theta^{-1}\mathcal{\tau}
  \end{array}\right),
\end{equation}
instead of equation (\ref{45}).

Further discussions about the fluid system obtained thus can be found in \cite{ZHYY}, including some choices of the freedoms.

\section{Long-time tendency and compatibility}
\setcounter{equation}{0}

In this section, we show that the CDF-based balance laws (\ref{32}) meet the long-time tendency and compatibility requirements, while the hyperbolicity and time-irreversibility have been shown in Section 3.

To begin with, we introduce a positive parameter $\varepsilon$ in the right-hand side of (\ref{32}):
\begin{equation}\label{51}
\partial_tU+\sum_{j=1}^3 \partial_{x_j}F_j(U)=\frac{Q(U)}{\varepsilon}.
\end{equation}
Notice that the first $n$ components of $Q(U)$ vanish. This parameter can be regarded as a relaxation time. A small $\varepsilon$ means that the dissipative variables $v$ evolve much faster than the conserved ones $u$. Namely, the time scale for $v$ to reach stationary
%, referred to as the relaxation time,
is much smaller than that for $u$. With the concept of relaxation time, a long-time can be simply defined as a time that is much longer than the relaxation time $\varepsilon$. Thus, the long-time tendency is equivalent to the zero relaxation-time limit where $\varepsilon$ tends to zero.\\

(5A). For small $\varepsilon$, the following result was established in \cite{Y1}. For sufficiently smooth initial data, the solution $U^\varepsilon=U^\varepsilon(x, t)$ to the CDF-based PDEs (\ref{51}) exists uniquely in an $\varepsilon$-independent time interval and has the expansion
$$
U^\varepsilon = \left( \begin{array}{l}
 u_0 \\
 v_0 \\
 \end{array} \right) + O(\varepsilon)
$$
in a certain Sobolev space. Here $u_0=u_0(x, t)$ and $v_0=c_0(x, t)$ solve the so-called equilibrium system
$$
q_v(u_0, v_0) = 0, \qquad u_{0t} + \sum\limits_{j=1}^3f_j(u_0,v_0)_{x_j} = 0.
$$
The expansion shows that, in the long time or when $\varepsilon$ is small, the CDF-based PDEs (\ref{32}) can be well replaced with the equilibrium system. Namely, the equilibrium system approximately describes the thermodynamic system under consideration.

For one-component fluid flows where the CDF-based balance laws are equations (\ref{11}) together with (\ref{45}), the equilibrium system is just the classical compressible Euler equations
$$
\begin{array}{rlrl}
\partial_t\rho  + \nabla\cdot(\rho\vecv) & = & 0, \\[3mm]
\partial_t(\rho \vecv) + \nabla\cdot(\rho \vecv\otimes \vecv +\pi\vecI) & = & 0, \\[3mm]
\partial_t(\rho E) + \nabla\cdot(\rho\vecv E +\pi\vecv ) & = & 0
\end{array}
$$
with the equation of state (\ref{44}):
\begin{equation}\label{52}
\pi = \pi(\nu, e) \equiv \frac{s_\nu(\nu, e, \vecw, {\bf C})}{s_e(\nu, e, \vecw, {\bf C})},
\end{equation}
where $\vecw=\vecw(\nu, e)$ and ${\bf C}={\bf C}(\nu, e)$ are obtained by solving
$$
s_{\vecw}(\nu, e, \vecw, {\bf C}) = 0, \quad s_{\bf C}(\nu, e, \vecw, {\bf C}) = 0.
$$
These algebraic equations are globally uniquely solvable thanks to the strict concavity of the generalized specific entropy $s = s(\nu, e, \vecw, {\bf C})$ in (\ref{43}).
For $\vecv=0$, we arrive at the equilibrium states characterized with the specific entropy
\begin{equation}\label{53}
s = s(\nu, e, \vecw(\nu, e), {\bf C}(\nu, e)).
\end{equation}
The concavity of this entropy function will be shown in Theorem 5.1 below. This verifies the long-time tendency to equilibrium. Remark that the equation of state (\ref{52}) or the equilibrium entropy (\ref{53}) place a constraint on the undetermined entropy function $s=s(\nu, e, \vecw, {\bf C})$. \\

(5B). To show the compatibility, we introduce the following transformation
\begin{equation} \label{54}
U=\left( \begin{array}{l}
 u \\
 v \\
 \end{array} \right) \longrightarrow
W\equiv \left( \begin{array}{l}
 u \\
 z \\
 \end{array} \right)=\left( \begin{array}{c}
 u \\
 \eta_{v}(u, v) \\
 \end{array} \right).
\end{equation}
Thanks to the strict concavity of $\eta(U)$,
the algebraic equation $z= \eta_{v}(u, v)$ can be globally and uniquely solved to obtain $v=v(u, z)$. Namely, the transformation is globally invertible. Under this
transformation, the system (\ref{51}) for smooth solutions can be
rewritten as
\begin{equation}\label{55}
W_t+\sum_{j=1}^d(D_U W)F_j(U)_{x_j}=\frac{1}{\varepsilon}\left( \begin{array}{c}
 0 \\
\eta_{vv}(U)M(U)z \\
 \end{array} \right)
\end{equation}
with $D_UW=\left[ {\begin{array}{*{20}{c}}
   I_n & 0  \\
   \eta_{vu}(U) & \eta_{vv}(U)  \\
\end{array}} \right]$ and $I_k$ the unit matrix of order $k$. For the CDF-based fluid model (\ref{41}) and (\ref{42}), this transformed system is
$$
\begin{array}{rl}
\partial_t\rho + \nabla\cdot(\rho\vecv) & = 0, \\[3mm]
  \partial_t\vecv + \vecv\cdot\nabla\vecv - \frac{s_{\vecw\nu}}{\rho}\nabla\cdot\vecv - \frac{s_{\vecw\vecw}}{\rho}\nabla\pi  & = \frac{s_{\vecw\vecw}\mathbf{M}\vecv}{\varepsilon\rho} .
\end{array}
$$

Transformation (\ref{54}) preserves the original conservation laws in (\ref{51}):
\begin{equation}\label{56}
u_t+\sum_{j=1}^df_j(u, z)_{x_j} = 0 .
\end{equation}
Here and below, the simplified notation $f_j(u, z)$ has been used to replace $f_j(u, v(u, z))$.
The $z$-equation in (\ref{55}) can be rewritten as
$$
z =\varepsilon M(U)^{-1}\eta_{vv}(U)^{-1} \Big[z_t +
\eta_{vu}(U)\sum_j f_j(u, z)_{x_j} + \eta_{vv}(U)\sum_jg_j(u, z)_{x_j}\Big].
$$
This indicates that $z=O(\varepsilon)$. Iterating the last equation yields
$$
z =\varepsilon \bar M(u)^{-1}\bar\eta_{vv}(u)^{-1}\sum_j\big[
\bar\eta_{vu}(u) f_{ju}(u, 0) + \bar\eta_{vv}(u)g_{ju}(u,
0)\big]u_{x_j}+ O(\varepsilon^2),
$$
where $\bar M(u)=M(u, v(u, 0))$ giving the meaning of the bar.
Moreover, $f_j(u, z)$ in (\ref{56}) can be expanded into
$$
f_j(u, z) = f_j(u, 0) + f_{jz}(u, 0)z+ O(\varepsilon^2).
$$

Substituting the two truncations above into (\ref{56}), we arrive at
the following second-order PDEs
\begin{equation}\label{57}
u_t+\sum_{j=1}^3f_j(u, 0)_{x_j}=\varepsilon
\sum_{j,k=1}^3(B^{jk}(u)u_{x_k})_{x_j}
\end{equation}
with
\begin{equation}\label{58}
B^{jk}(u) = -f_{jz}(u,0)\bar M(u)^{-1}\bar
\eta_{vv}(u)^{-1}\big[ \bar\eta_{vu}(u) f_{ku}(u, 0) +
\bar\eta_{vv}(u)g_{ku}(u, 0)\big].
\end{equation}
This procedure in deriving (\ref{57}) from (\ref{55}) is called {\it Maxwell iteration} \cite{Muller_book}
and gives the exactly same result as the Chapman-Enskog expansion does \cite{Y3}.

For the CDF-based heat conduction model above, the corresponding second-order system (\ref{57}) reads as
$$
\partial_t \rho + \nabla\cdot(\rho\vecv) = 0, \quad
\vecv = - \varepsilon\mathbf{M}(\rho,\rho\vecw)^{-1}\nabla\pi ,
$$
where $\vecw =\vecw(\nu)$ is obtained by solving $s_{\vecw}(\nu, \vecw)=0$. Recall that $\pi=s_\nu(\nu, \vecw)$. This is the classical porous medium equation
$$
\partial_t \rho = \varepsilon\nabla\cdot(\rho^m\nabla \rho),
$$
provided that the dissipation matrix and entropy are taken as, for example,
$$
\mathbf{M}(\rho,\rho\vecw)=2\rho^{2-m}\vecI \quad \mbox{and} \quad s(\nu,\vecw) =- \nu^{-1} + h(\vecw)
$$
with $m$ a real number and $h(\vecw)$ strictly concave.

\begin{rem}
From such choices of $\mathbf{M}(\rho,\rho\vecw)$ and $\eta(\rho, \rho\vecw)=\rho s(1/\rho, \vecw)$, we see that the compatibility and long-time tendency analysis provide significant constraints on the freedoms of CDF: the strict concave entropy function and positive-definite dissipation matrix. A similar analysis for the CDF-based one-component fluid system (\ref{11}) and (\ref{45}) (or (\ref{46})) can also be done but is a bit tedious.
\end{rem}

Denote by $W^\varepsilon_h\equiv (u^\varepsilon_h, z^\varepsilon_h)$ and $u^\varepsilon_p$ the solutions to
systems (\ref{55}) and (\ref{57}), respectively. It was proved in \cite{YangY} for sufficiently smooth initial data that the expansion
$$
u^{\varepsilon}_h -u^{\varepsilon}_p = O(\varepsilon^2)
$$
holds in a certain Sobolev space. This result indicates that the second-order PDEs (\ref{57}) is a good approximation to the CDF-based first-order system (\ref{51}) when the dissipative variables evolve much faster than the conserved ones. Consequently, the Maxwell iteration makes up for the shortcoming of CDF that only first-order PDEs can be obtained. \\

(5C). The CDF-based second-order PDEs (\ref{57}) possess the following important structural properties.

\begin{thm} (\cite{YangY}) Set $\hat\eta(u, z) := \eta(u, v(u, z))$. Then $\hat\eta (u):=\hat\eta(u, 0)$ is strictly concave with respect to $u$
under consideration, $\hat\eta_{uu}(u)f_{ju}(u, 0)$ is symmetric, and the CDF-based second-order PDEs (\ref{57}) can be rewritten as
\begin{equation}\label{59}
u_t+\sum_{j=1}^3f_j(u, 0)_{x_j}=-\varepsilon
\sum_{j,k=1}^3({\tilde B}^{jk}(u)({\hat\eta}_u(u))_{x_k})_{x_j}
\end{equation}
with
\begin{equation}\label{510}
{\tilde B}^{jk}(u) = f_{jz}(u, 0)\bar M(u)^{-1}f_{kz}(u,0)^T.
\end{equation}
\end{thm}

\begin{proof}
First of all, the symmetric positive-definite matrix
$$
A_0(W) := -(D_WU)^T\eta_{UU}(U(W))D_WU
$$
is a symmetrizer in the sense of Friedrichs \cite{Dafermos} for the quasilinear system (\ref{55}). Namely, $A_0(W)D_U WF_j(U(W))_W$ is symmetric. Since
$$
D_WU= (D_UW)^{-1}=\left[ {\begin{array}{*{20}{c}}
I_n & 0  \\
 -\eta_{vv}^{-1}\eta_{vu} & \eta_{vv}^{-1}  \\
\end{array}} \right],
$$
it is direct to verify that the symmetrizer is a block-diagonal matrix with the above partition.

Notice that $\hat\eta(W) = \hat\eta(u, z).$ We compute
$$
\hat \eta_{WW}(u, 0) =(D_WU)^T\eta_{UU}D_WU = -A_0(u, 0)
$$
at equilibrium where $q=\eta_v(U)=0$. This indicates the negative-definitenss of $\hat\eta_{uu}(u, 0)$ and thereby the strict concavity of $\hat\eta(u, 0)$.
In addition, since $A_0(W)$ is block-diagonal and $A_0(W)D_U WF_j(U(W))_W$ is symmetric, it follows immediately from  the expression of $D_U WF_j(U(W))$ that $\hat\eta_{uu}(u, 0)f_{ju}(u, 0)$ is symmetric and
\begin{equation}\label{511}
\hat\eta_{zz}(u, 0)[\bar\eta_{vu}(u)f_{ju}(u, 0)+\bar\eta_{vv}(u)g_{ju}(u, 0)]=[\hat\eta_{uu}(u, 0)f_{jz}(u, 0)]^T.
\end{equation}
Moreover, a direct calculation shows that
$$
\hat \eta_{zz}(u, 0)= \eta_{vv}^{-1}(U)\big|_{z=0} =\bar\eta_{vv}(u)^{-1} .
$$
This together with (\ref{511}) leads to the following relation
$$
B^{jk}(u) = -f_{jz}(u, 0)\bar M(u)^{-1}\big{[}\hat\eta_{uu}(u, 0)f_{kz}(u,0)\big{]}^T.
$$
The theorem follows from this relation.
\end{proof}

Theorem 5.1 exposes the gradient structure \cite{HMie} of the CDF-based second-order PDEs (\ref{59}). When the dissipation matrix $\bar M(u)$ is symmetric, it also indicates that
$$
{\tilde B}^{jk}(u) = ({\tilde B}^{kj}(u))^T.
$$
Namely, the $(3n)\times(3n)$-matrix
$$
 \left[ {\begin{array}{*{20}{c}}
{\tilde B}^{11}(u) & {\tilde B}^{12}(u) & {\tilde B}^{13}(u) \\
{\tilde B}^{21}(u) & {\tilde B}^{22}(u) & {\tilde B}^{23}(u) \\
{\tilde B}^{31}(u) & {\tilde B}^{32}(u) & {\tilde B}^{33}(u)
\end{array}} \right]
$$
is symmetric. Such a matrix corresponds to a nonlinear reciprocal relation between the thermodynamic flux and (entropic) force $((\hat\eta_{u}(u))_{x_1}, (\hat\eta_{u}(u))_{x_2}, (\hat\eta_{u}(u))_{x_3})$ for vectorial or tensorial irreversible processes \cite{GM}.

Furthermore, we have

\begin{thm} (\cite{YangY})
The CDF-based second-order PDEs (\ref{59}) is strongly dissipative \cite{Dafermos, serre}. Namely, there exists $c_0=c_0(u)>0$ such that
$$
\sum\limits_{j,k=1}^3\xi_j^T{\tilde B}^{jk}(u)\xi_k\geq
c_0(u)\sum\limits_{j=1}^3\Big|\sum\limits_{k=1}^3{\tilde B}^{jk}(u)\xi_k\Big|^2
$$
holds for any $\xi_j\in{\bf R}^n (j=1, 2, 3)$.
\end{thm}

\begin{proof}
With (\ref{510}), we can write
$$
\begin{array}{rl}
\sum\limits_{j,k=1}^3\xi_j^T{\tilde B}^{jk}(u)\xi_k&=\sum\limits_{j,k=1}^3\xi_j^Tf_{jz}(u,0)\bar M(u)^{-1}f_{kz}^T(u,0)\xi_k\\
&=\Big{(}\sum\limits_{j=1}^3f_{jz}^T(u,0)\xi_j\Big{)}^T\bar M(u)^{-1}\Big{(}\sum\limits_{k=1}^3 f_{kz}^T(u,0)\xi_k\Big{)}\\
&\geq \delta_1(u)\Big|\sum\limits_{k=1}^3f_{kz}^T(u,0)\xi_k\Big|^2
\end{array}
$$
with $\delta_1(u)$ the largest eigenvalue of the symmetric positive-definite matrix $\bar M(u)$.
On the other hand, there exists $\delta_2(u)>0$ such that
$$
\begin{array}{rl}
\sum\limits_{j=1}^3|\sum\limits_{k=1}^3{\tilde B}^{jk}(u)\xi_k|^2&
=\sum\limits_{j=1}^3\Big|\sum\limits_{k=1}^3f_{jz}(u,0)\bar M(u)^{-1}
f_{kz}^T(u,0)\xi_k\Big|^2\\
&=\sum\limits_{j=1}^3\Big|f_{jz}(u,0)\bar M(u)^{-1}\Big{(}\sum\limits_{k=1}^3f_{kz}^T(u,0)\xi_k\Big{)}\Big|^2\\
&\leq \delta_2(u)\Big|\sum\limits_{k=1}^3f_{kz}^T(u,0)\xi_k\Big|^2.
\end{array}
$$
Hence the proof is complete by taking $c_0(u)=\delta_1(u)/\delta_2(u)$.
\end{proof}

It is worthwhile to comment that Theorems 5.1 and 5.2 correspond to the conservation-dissipation principle and the entropy-production inequality (\ref{35}). They indicate that the Maxwell iteration preserves both the principle and the inequality. The interested reader is referred to \cite{Y3} for other properties preserved by the iteration.

\section{Further comments}
\setcounter{equation}{0}

This section contains some further comments on CDF, including its limitations.

First of all, CDF assumes that certain conservation laws are given {\it a priori}. This particularly means that CDF accommodates other possibly existing conservation laws rather than only that of energy.
Notice that conservation laws usually root in first principles and therefore should be respected fully.

Numerical methods respecting conservation laws have been developed in \cite{luc, xu1, xu2} for rarefied gas flows. The idea of CDF has also been used in design of numerical methods \cite{luc} based on the Boltzmann equations. Moreover, CDF-based models are nothing but traditional PDEs. They are very amenable to existing numerics.

CDF-based PDEs have the conservative form in (\ref{32}), which allows a convenient definition of weak solutions. Let $F_j=F_j(U) (j=1, 2, 3)$ and $Q=Q(U)$ be continuous with respect to $U$. A vector-valued locally bounded function $U=U(x, t)$ is called a weak solution to balance laws (\ref{32}) if
the following equality
$$
\int_0^\infty\int_{x\in\mathbb{R}^3}\Big[ U\phi_t + \sum_{j=1}^3F_j(U)\phi_{x_j} + Q(U)\phi\Big]dxdt =0
$$
holds for any smooth function $\phi=\phi(x, t)$ with compact support \cite{Dafermos, serre}. This is a natural generalization of the definition of classical (or continuously differentiable) solutions. It fully utilizes the conservative form of the flux term $F_j(U)_{x_j}$. If the flux term is replaced with a term of form $A_j(U)U_{x_j}$, how to define non-classical solutions has been an open problem \cite{Dafermos, serre}. It should be mentioned that weak solutions cannot be evaded in studying interface or other problems where discontinuous solutions like shock waves arise.

Another advantage of the conservative form is related to the hyperbolicity. For a system of conservation laws with a concave entropy, it is well-known that the Hessian matrix of the entropy symmetrizes the system in the sense of Friedrichs \cite{FL_71, Dafermos}. Consequently, the system is symmetrizable hyperbolic. Thus, if a system of conservation laws has a concave entropy, its hyperbolicity and thereby well-posedness follows automatically.

When nonconservative terms exist, the above conclusion does not hold in general.
To see this, we consider the GENERIC-based compressible model for complex fluids proposed in \cite{BET_book}.
The model uses an evolution equation for a configuration tensor and takes the stress as a function of the latter.
It obeys both the principle of material indifference and the second law of thermodynamics.
In the isothermal case, the model reads as
\begin{equation}\label{61}
  \begin{array}{rl}
  \partial_t \rho + \nabla \cdot (\rho \mathbf{v}) & = 0, \\[3mm]
  \partial_t (\rho \mathbf{v}) + \nabla \cdot (\rho \mathbf{v}\otimes\mathbf{v} + p\mathbf{I}) + \nabla \cdot \pi & =0, \\[3mm]
  \partial_t \mathbf{c} + \mathbf{v} \cdot \nabla \mathbf{c} + (\nabla \mathbf{v})  \mathbf{c }+ \mathbf{c}  (\nabla \mathbf{v})^T -(\nabla \mathbf{v} +(\nabla \mathbf{v})^T ) & = -\frac{\mathbf{c}}{\tau}. \label{11-3}
\end{array}
\end{equation}
Here the first two equations are the standard conservation laws (of mass and momentum) and the third one is a constitutive equation for the configuration tensor $\mathbf{c}$.
In (\ref{61}),
$p=p(\rho)$ is the fluid pressure,
$\pi=2\mathbf{c} (\mathbf{c}-\mathbf{I})-\frac{1}{2}\mbox{tr}(\mathbf{c}^2)\mathbf{I}$ is the stress tensor with $\mbox{tr}(\mathbf{c}^2) = \sum_{i,j}c_{ij}c_{ji}$, $\nabla \mathbf{v}=[ \partial_{x_j} v_i]$, and $\tau$ is a relaxation time related to the viscosity of fluid. Notice that this model consists of only first-order PDEs and contains non-conservative terms like $(\nabla \mathbf{v})\mathbf{c }$ in (\ref{11-3}) due to the Oldroyd derivative.

For (\ref{61}), it is easy to check that the following function from \cite{BET_book, Y6}:
$$
  \eta = \eta(\rho,\rho \mathbf{v} ,\mathbf{c}) = -\rho \int^\rho \frac{p(z)}{z^2} dz - \frac{(\rho \mathbf{v})^2}{2\rho} - \frac{1}{2} \mbox{tr}(\mathbf{c}^2)
  $$
is strictly concave with respect to $U$ under the physically reasonable assumption $p_{\rho} \ge 0$. Moreover,
with the equations in (\ref{61}) we can compute
$$
\begin{array}{rlrl}
  \eta_t &=& \eta_{\rho} \rho_t + \eta_{\rho \mathbf{v}} \cdot \mathbf{v}_t + \eta_\mathbf{c} : \mathbf{c}_t \nonumber\\
%  =-S_{\rho} {\rho}_t \nabla \cdot (\rho \mathbf{v}) - S_{\rho \mathbf{v}} \cdot ( \nabla \cdot (\rho \mathbf{v} \mathbf{v}) + \nabla p + \nabla \cdot \pi) - S_{\mathbf{c}} : (\mathbf{v} \cdot \nabla \mathbf{c} + (\nabla \mathbf{v})^T \cdot \mathbf{c} + \mathbf{c} \cdot \nabla \mathbf{v} - (\nabla \mathbf{v} +(\nabla \mathbf{v})^T)) -\frac{\mathbf{c}:\mathbf{c}}{\tau} \\
 % =-\nabla \cdot (S\mathbf{v}+p \mathbf{v}+ \mathbf{\pi} \cdot \mathbf{v}) + (\rho S_\rho + \rho \mathbf{v}^2-p-\phi(\rho)-\frac{1}{2}\rho \mathbf{v}^2) \nabla \cdot \mathbf{v} + (\mathbf{c}\cdot 2\mathbf{c}-2\mathbf{c}-\frac{1}{2} \mathbf{c}:\mathbf{c}\mathbf{I} - \mathbf{\pi}):\nabla \mathbf{v} - \frac{1}{\tau}\mathbf{c}:\mathbf{c} \\
  &=& -\nabla \cdot (\eta\mathbf{v}-p \mathbf{v}- \mathbf{\pi}\mathbf{v}) + \frac{1}{\tau}\mathbf{c}:\mathbf{c} \\
  &\equiv& -\nabla \cdot J(U) + \sigma .\nonumber
\end{array}
$$
Namely, $\eta=\eta(\rho,\rho \mathbf{v}, \mathbf{c})$ defined above is a local entropy with the entropy flux $J(U)$ and the local entropy production rate $\sigma$. The non-negativeness of $\sigma=\frac{1}{\tau}\mathbf{c}:\mathbf{c}\geq0$ guarantees the second law of thermodynamics.

However, it was shown in \cite{HuoY} that the Hessian matrix of the entropy function $\eta=\eta(\rho,\rho \mathbf{v}, \mathbf{c})$ does not symmetrize the system of first-order PDEs in (\ref{61}) in both one- and two-dimensional cases. In the two-dimensional case, the symmetrizable hyperbolicity of the system (\ref{61}) remains unclear. Thus, the well-posedness (hyperbolicity) of the GENERIC-based PDEs (\ref{61}) needs to be examined.

Although CDF has many advantages, it has its limitations. For example, CDF-based models do not contain constitutive equations with nonlocal terms. In addition, CDF's conservative form is inconsistent to the principle of material indifference, which lead to the typical non-conservative term $(\nabla \mathbf{v})\mathbf{c}$ in (\ref{11-3}). So far, it is not clear how to remedy the CDF so that the principle of material indifference is incorporated. A possibly reasonable alternative to the conservation-dissipation principle would be the structural stability conditions proposed in \cite{Y1, Y2} for hyperbolic systems of quasilinear PDEs with relaxation. The structural stability conditions are weaker than the conservation-dissipation principle but can also guarantee the four fundamental requirements.

At this point, let me mention the works of the Godunov school (currently in \cite{PR4} called the Symmetric Hyperbolic Thermodynamically Compatible framework --- SHTC). It was suggested explicitly in \cite{PR2} to model dissipative processes by lower order terms of the associated hyperbolic systems of equations. In the past several decades, various special systems of first-order PDEs, including non-conservative evolution equations, for specific physical processes have been constructed. Examples inculde a system of 10 equations in \cite{PR1} for nonlinear elastic processes, a system of 4 equations in \cite{PR2} for heat transfer, a system of 14 equations in \cite{PR5} for viscous fluid flows, etc. However, there seems no a reference where the SHTC framework
%for systems of $n$ equations
is defined clearly (as clearly as GENERIC \cite{BET_book} or as clearly as our CDF \cite{ZHYY}). On the other hand, it is known that mere hyperbolicity is not enough to prevent inherently instability of a PDE system. An example is the well-known BISQ model in geophysics. It was shown in \cite{Jiawei} that the hyperbolic system allows time-exponentially exploding solutions. Therefore, the crucial point in this context is how to couple the lower-order terms with the hyperbolic parts. This point seems missing in the papers of the school. Consequently, it remains unclear whether or not the special systems constructed by the school meet the
%four
fundamental requirements. And it would be interesting to examine whether or not the special systems satisfy the conservation-dissipation principle or structural stability conditions \cite{Y1, Y2}, whose necessity was expounded in \cite{Y2}.}

Finally, I mention a possible combination of CDF with machine learning. CDF provides a unified framework for modeling irreversible processes endowed with conservation laws. It has two problem-dependent freedoms. For specific problems, the freedoms may be fitted through machine leaning. A seemingly feasible project would be to fix the freedoms for rarefied gas dynamics. As is well-known, the Boltzmann equation is a well-accepted model for rarefied gas dynamics but computationally costly. A reliable and unified macroscopic model is desirable. For that purpose, various moment closure systems have been developed and tested (see, e.g., \cite{Muller_book}). A frequent question is how many moments are needed. The most possible answer would be 14, once  the freedoms $\eta(U)$ and $M(U)$ (46 functions of $U\in {\bf R}^{14}$) in the one-component fluid system (\ref{11}) and (\ref{46}) are fixed. Remark that a great number of reliable solutions (data) of the Boltzmann equation are available, which have been obtained through various numerical methods or experiments. Those data may be used to fix the freedoms, with the aim to obtain a unified macroscopic model for rarefied gas flows.

\section{Summary}
\setcounter{equation}{0}

This section collects the main conclusions of this paper. I propose four fundamental requirements for establishing evolution PDEs modeling irreversible processes. These requirements are the observability of irreversible processes, irreversibility, long-time tendency to equilibrium, and compatibility with the possibly existing well-validated theories. I show that the CDF-based PDEs meet all the four requirements.

In showing the compatibility, the Maxwell iteration is applied to the CDF-based first-order PDEs and, consequently,  second-order PDEs are derived systematically. The analysis not only reveals significant constraints on the freedom of CDF, it also remedies the shortcoming of CDF that only first-order PDEs can be derived. It is proved that the derived second-order PDEs preserve the gradient structure and strong dissipativeness of the first-order ones.

The strong dissipativeness is characterized with a refined formulation of the second law of thermodynamics, which is the entropy-production inequality (\ref{35}) and Theorem 5.2 for the first-order and second-order PDEs, respectively. The inequality follows from the conservation-dissipation principle---the basis of CDF. On the other hand, the gradient structure corresponds to nonlinear Onsager reciprocal relations for scalar processes (described with the first-order PDEs) and for vectorial or tensorial processes (described with the second-order PDEs).

\end{document}